\documentclass{article}

\usepackage{amsmath,leftindex,amsfonts,amssymb,amsthm,graphicx,float}

\newtheorem{theorem}{Theorem}[section]

\newtheorem{proposition}{Proposition}[section]

\newtheorem{lemma}{Lemma}[section]
\newtheorem{corollary}{Corollary}[section]
\theoremstyle{definition}

\renewcommand{\Re}{\operatorname{Re}}

\def\H{{\mathcal{H}}}

\def\RR{\mathbb{R}}

\def\NN{\mathbb{N}}

\def\RA1{\mathfrak{R}_{A'}}
\def\NA1{\mathfrak{N}_{A'}}

\author{Bogdan D. Djordjevi\'c\footnote{Institutions: Mathematical Institute of the Serbian Academy of Sciences and Arts, (MI SASA), Serbia. Visiting researcher at the Institute of
Mathematics and Informatics, Bulgarian Academy of Sciences (IMI-BAS), Bulgaria, and the Faculty of Mathematics and Informatics, Sofia University ``St. Kliment Ohridski" (FMI-SU), Bulgaria.\\
email: bogdan.djordjevic@turing.mi.sanu.ac.rs; bogdan.djordjevic93@gmail.com}, Nikolay A. Ivanov\footnote{Institution: Faculty of Mathematics and Informatics, Sofia University ``St. Kliment Ohridski" (FMI-SU), Bulgaria.\\ email: nivanov@fmi.uni-sofia.bg}}
\date{}

\begin{document}
	
	\title{Variants of the Quantum Phase Operator for the Harmonic Oscillator}
	
	\maketitle
	
	\begin{abstract} 
		We introduce and study quantum phase operators associated with the Quantum Harmonic Oscillator (QHO). We show that these operators are trace-class perturbations of the Susskind–Glogower operators and examine their mathematical and physical properties. The construction is motivated by the physically relevant two-phase case. 
	\end{abstract}
	
	
	{ Keyphrases: Quantum Harmonic Oscillator, Spectral Theory, Quantum Phase Operator.}

\section{Introduction}
Physical relevance of the quantum phase operators has been long recognized and extensively discussed in the literature; see, for example, \cite{NFM91,NFM92}. The notion of phase plays a fundamental role in quantum theory, particularly in systems where only relative phases (phase differences) possess physical significance, reflecting the invariance of quantum states under global phase transformations.

On this occasion, our main concern is mathematical definition and analytic treatment of phase operators occurring in the Quantum Harmonic Oscillator (QHO). In this scenario, a distinctive feature of the phase observable is its intimate connection with the excitation level of the system. In particular, the concept of phase becomes problematic in the vacuum state, for which no physically meaningful phase can be assigned. Conversely, both theoretical considerations and experimental observations indicate that phase can be determined with increasing precision as the number of photons increases. This behavior is especially evident in high-intensity laser fields, where large occupation numbers allow for a more accurate characterization of phase-related properties. 

Following the framework developed in \cite{LSJG}, in this paper we define and study a variant for quantum phase operators that arise from the correspondence between Bose–Einstein statistics and the quantum harmonic oscillator (QHO). In Section \ref{onephase} we develop mathematical foundation for the phase operators, examine their commutator and spectral properties. The theoretical motivation and heuristics of our considerations are then discussed in detail in Section \ref{two phases}. 
	\subsection{Theoretical setup} 
    Recall that (see e.g., \cite{W15}) for the QHO the hamiltonian of the system is
	\begin{equation}\label{H} H=\frac{p^2}{2m}+\frac{m}{2}\omega^2x^2,\end{equation}
	where the position operator $x$ and momentum operator $p$ satisfy the well-known canonical commutation relation:
	\begin{equation}\label{xppx}[x,p]=i \hbar I.\end{equation}
	The following observation follows immediately:
	\begin{proposition} The following relations are true for $H$, $x$, and $p$:
		\begin{equation}\label{Hx} [H,x]=-i\frac{\hbar}{m}p,
		\end{equation}
		\begin{equation}
			\label{Hp}
			[H,p]=i\hbar \omega^2mx.
		\end{equation}
	\end{proposition}
	\begin{proof} from \eqref{xppx} it follows that 
		$$xp^2-p^2x=(xp)p-p(px)=(px+i\hbar)p-p(xp-i\hbar)=2i\hbar p,$$
		therefore
		$$[H,x]=\frac{-1}{2m}(xp^2-p^2x)=\frac{-i\hbar }{m}p.$$
		Similarly, 
		$$x^2p-px^2=x(xp)-(px)x=x(px+i\hbar)-(xp-i\hbar)x=2i\hbar x,$$
		therefore
		$$[H,p]=\frac{m\omega^2}{2}(x^2p-px^2)=i\hbar \omega^2 mx.$$
	\end{proof}
	
It is a known fact that (see \cite{GT}) $H$ from \eqref{H} is self-adjoint, bounded from below, and has a purely point spectrum $\sigma(H)$, comprising out of eigenvalues of the half-integer form:
	\begin{equation}\label{sigmaH}
		\sigma(H)=\left\{\hbar\omega\left(n+\frac{1}{2}\right): n\in\NN_0\right\},
	\end{equation}
	where each eigenvalue is simple, and the corresponding eigen-function is the $n-$th Hermite function $\Psi_n(x)$:
	\begin{equation}\label{eigenH}
		H\psi(x)=\hbar\omega\left(n+\frac{1}{2}\right)\psi(x)\Leftrightarrow \psi(x)=\Psi_n(x), \quad n\in\NN_0,\end{equation}
	with
	\begin{equation}\label{Psin}\Psi_n(x)=\frac{1}{\sqrt{2^n n!}}\left(\frac{m\omega}{\pi\hbar}\right)^{1/4}e^{-\frac{m\omega x^2}{2\hbar}}H_n\left(\sqrt{\frac{m\omega}{\hbar}}x\right),\end{equation}
    where $H_n$ is the physicists' Hermite polynomial, 
\begin{equation}\label{Hermiten}
H_n(z)=(-1)^ne^{z^2}\frac{d^n}{dz^n}\left(e^{-z^2}\right),\quad n\in\NN_0.
\end{equation}
	Combined, by writing $\Psi_n$ as $n$, the eigen-problem \eqref{sigmaH}-\eqref{eigenH} reads
	\begin{equation}\label{Hn}
		H|n\rangle=\hbar \omega\left(n+1/2\right)|n\rangle, \quad, n\in\NN_0.
	\end{equation}
	Accordingly, the ladder operators are defined as:
	\begin{equation}\label{a}
		a=\frac{1}{\sqrt{2m\omega\hbar }}\left(-i\hbar \frac{\partial}{\partial x}-im\omega x\right)=\frac{1}{\sqrt{2m\omega\hbar }}\left(p-im\omega x\right),
	\end{equation}
	\begin{equation}\label{adag}
		a^\dagger=\frac{1}{\sqrt{2m\omega\hbar }}\left(-i\hbar \frac{\partial}{\partial x}+im\omega x\right)=\frac{1}{\sqrt{2m\omega\hbar }}\left(p+im\omega x\right),
	\end{equation}
	which satisfy the canonical commutation relations:
	\begin{equation}\label{aadag} [a,a^\dagger]=I\end{equation}
	\begin{equation}\label{Ha}[H,a]=-\hbar \omega a,\quad [H,a^\dagger]=\hbar \omega a^\dagger.\end{equation}
	Ergo, $a$ is a lowering operator while $a^\dagger$ is a raising operator, in the sense that:
	\begin{equation}\label{an}
		\begin{cases}a| n\rangle=\sqrt{n}|n-1\rangle, \ n\geq 1; \quad a|0\rangle=0;\\a^\dagger|n=\sqrt{n+1}|n+1\rangle.\end{cases}
	\end{equation}
	The operators $a$ and $a^\dagger$ produce the number operator $N$, 
	\begin{equation}
		\label{N}
		N=a^\dagger a
	\end{equation}
	which is self-adjoint, bounded from below, and
	\begin{equation}\label{NH}
		H=\hbar \omega\left(N+\frac{1}{2}I\right)=\hbar \omega\left(a^\dagger a+\frac{1}{2}I\right)=\hbar \omega\left(aa^\dagger-\frac{1}{2}I\right).\end{equation}
	Finally, by adding and subtracting \eqref{a} from \eqref{adag} respectively, we get
	\begin{equation}\label{xpaadag}
		p=\sqrt{\frac{m\omega\hbar}{2}}(a+a^\dagger),\quad x=i\sqrt{\frac{\hbar}{2m\omega}}\left(a-a^\dagger\right).
	\end{equation}
	The latter shows that $a+a^\dagger$ and $i(a-a^\dagger)$ are self-adjoint operators. 
	\par 
	The Quantum phase operator was considered and theoreticized in \cite{LSJG}. Later on it was studied by many 
	scholars, e.g. \cite{CN65} and \cite{BP90}. 
	\par

	\subsubsection{Positive square root and its inverse}
	Since $H$ is a positive operator, there exists a unique positive square root $H^{1/2}$, defined on the same domain as the operator $H$ is. This is guaranteed by the Kato theorem below:
	\begin{theorem}[\cite{Kato}]
		Let $E: D_E\mapsto\H$ be a strictly accretive linear operator, $D_E\subset\H$. Then there exists a unique strictly accretive square root $E^{1/2}$ defined in $D_E$, such that $\forall u\in D_E$: $\left(E^{1/2}\right)^2u=Eu$, and
		\begin{equation}\label{sqrtE}
			E^{1/2}u=\frac{1}{\pi}\int_0^\infty \left(\sqrt{t}\right)^{-1}E\left(t+E\right)^{-1}udt.
		\end{equation}
	\end{theorem}
	Consequently, the operator $H^{1/2}$ as such exists, defined in the same set $D_H$ as $H$, and obeys to the formula \eqref{sqrtE} above:
	\begin{equation}\label{sqrtH}
		H^{1/2}u=\frac{1}{\pi}\int_0^\infty \left(\sqrt{t}\right)^{-1}H\left(t+H\right)^{-1}udt.
	\end{equation}
	
	Recall that Hermite polynomials comprise the eigen-functions for $H$ that form its corresponding egienstates, denoted as \eqref{Hn}. Then by the spectral mapping theorem we have
	\begin{equation}\label{}
		{H}^{1/2}|n\rangle=\sqrt{\omega\hbar \left(n+1/2\right)}|n\rangle,\quad n\in\NN_0.
	\end{equation}
	Therefore, its spectrum is also bounded from below, and by the closed graph theorem, the operator $H^{1/2}$ has a bounded inverse, and due to associativity one has that
	$$\left(H^{1/2}\right)^{-1}=\left(H^{-1}\right)^{1/2},$$
	where $H^{-1}$ is the bounded inverse of $H$ and $\left(H^{-1}\right)^{1/2}$ is the unique positive square root of $H^{-1}$. In that sense, the operator
	\begin{equation}\label{H-1/2}
		H^{-1/2}:=\left(H^{1/2}\right)^{-1}=\left(H^{-1}\right)^{1/2}\end{equation}
	is well-defined bounded linear operator, and by the spectral mapping theorem yet again it follows that
	\begin{equation}\label{H-1/2n}
		H^{-1/2}|n\rangle=\frac{1}{\sqrt{\omega\hbar \left(n+1/2\right)}}|n\rangle, \quad n\in\NN_0.
	\end{equation}
	Therefore, throughout the remaining text we will understand that $H^{-1/2}$ denotes the bounded inverse of the positive square root of the given Hamiltonian, satisfying \eqref{H-1/2}--\eqref{H-1/2n}.
	\par

	\section{One oscillator phase}\label{onephase}
	\subsection{Exact Sine and Cosine Operators}
	Ideally, one would want to define time-like operators $T_0$ and $T_0^\dagger$ in the following manner:
	\begin{equation}\label{TTdag}
		T_0:=\frac{\sqrt{\hbar \omega}}{2}\left(aH^{-1/2}+H^{-1/2}a\right),\quad T_0^\dagger:=\frac{\sqrt{\hbar \omega}}{2}\left(a^\dagger H^{-1/2}+ H^{-1/2}a^\dagger\right).
	\end{equation} 
	This is well-motivated by the two phase case in Section \ref{two phases}. 
	\par
	By the virtue of \eqref{H-1/2} and \eqref{H-1/2n}, the operators $T_0$ and $T_0^\dagger$ are well-defined and bounded. Additionally, by \eqref{TTdag}, they preserve the lowering-raising properties inherited from the operators $a$ and $a^\dagger$, respectively:
	\begin{lemma} The operators $T_0$ and $T_0^\dagger$ satisfy 
		\begin{equation}\label{HT0}
			[H,T_0]=-\hbar \omega T_0,\quad [H,T_0^\dagger]=\hbar \omega T_0^\dagger.
		\end{equation}
	\end{lemma}
	\begin{proof} Proof follows directly from \eqref{Ha} and a simple algebraic property:
		$$[H,aH^{-1/2}+H^{-1/2}a]=[H,a]H^{-1/2}+H^{-1/2}[H,a].$$ 
		The same procedure goes when we write $a^\dagger$ instead of $a$.
	\end{proof}
	For convenience, we are going to write 
	\begin{equation}
		\label{b} b:=\sqrt{\hbar\omega}a, \quad b^\dagger:=\sqrt{\hbar\omega}a^\dagger.
	\end{equation}
	That way the operators $T_0$ and $T_0^\dagger$ read 
	\begin{equation}
		\label{T0b} T_0=\frac{1}{2}\left(bH^{-1/2}+H^{-1/2}b\right),\quad T_0^\dagger=\frac{1}{2}\left(b^\dagger H^{-1/2}+H^{-1/2}b^\dagger\right).
	\end{equation}
	The following lemma will be useful for the later text:
	\begin{lemma}\label{bproperties} For the operator $b$ introduced in \eqref{b}, the following are true:
		\begin{itemize}
			\item[(a)]\begin{equation}
				\label{Hb}
				[H,b]=-\omega\hbar b ,\quad [H,b^\dagger]=\omega\hbar b^\dagger,\quad [b,b^\dagger]=\hbar\omega I.
			\end{equation}
			\item[(b)] \begin{equation}\label{H-1b}[H^{-1},b]=\hbar\omega H^{-1}bH^{-1},\quad [H^{-1},b^\dagger]=-\hbar\omega H^{-1}b^\dagger H^{-1}.
			\end{equation}
			Consequently, both $[H^{-1},b]$ and $[H^{-1},b^{\dagger}]$ are bounded operators.
			\item[(c)]\begin{equation}\label{H-1/2b}
				[H^{-1/2},b]=\hbar\omega \bar{b},\quad [H^{-1/2},b^\dagger]=-\hbar \omega \bar{b}^\dagger,
			\end{equation}
			where $\bar{b}$ and $\bar{b}^\dagger$ are respectively the unique bounded solutions to 
			\begin{equation}
				\label{barb}
				H^{-1/2}\bar{b}+\bar{b}H^{-1/2}=-H^{-1}bH^{-1},\quad H^{-1/2}\bar{b}^\dagger+\bar{b}^\dagger H^{-1/2}=-H^{-1}b^\dagger H^{-1}.
			\end{equation}
			Consequently, the commutators $[H^{-1/2},b]$ and $[H^{-1/2},b^\dagger]$ are bounded.
		\end{itemize}
	\end{lemma}
	\begin{proof} The relations \eqref{Hb} follow immediately from \eqref{aadag} and \eqref{Ha}. The relations \eqref{H-1b} follow from the commutator property:
		$$[H^{-1},b]=H^{-1}(bH-Hb)H^{-1}=-H^{-1}[H,b]H^{-1}=\omega\hbar H^{-1}bH^{-1},$$
		and similarly for $b^\dagger$. In order to prove \eqref{H-1/2b}, we first conduct a similar transform as the previous one,
		$$\begin{aligned}
			[b,H^{-1/2}]=&H^{-1/2}(H^{1/2}b-bH^{1/2})H^{-1/2}=\\
			=&H^{1/2}\left(H^{-1/2}bH^{-1/2}\right)-\left(H^{-1/2}bH^{-1/2}\right)H^{1/2}.
		\end{aligned}$$
		Note that $H^{-1/2}bH^{-1/2}$ is a bounded operator. That being said, the formula \eqref{sqrtH} applies to $H^{1/2}$ and
		\begin{eqnarray}\label{H-12bcomp}
			\begin{aligned}
				&[b,H^{-1/2}]=\\
				&H^{1/2}\left(H^{-1/2}bH^{-1/2}\right)-\left(H^{-1/2}bH^{-1/2}\right)H^{1/2}=\\
				&\frac{1}{\pi}\int_0^\infty t^{-1/2}H\left(t+H\right)^{-1}\left(H^{-1/2}bH^{-1/2}\right)udt-\\
				-&\frac{1}{\pi}\int_0^\infty t^{-1/2}\left(H^{-1/2}bH^{-1/2}\right)H\left(t+H\right)^{-1}udt,
			\end{aligned}
		\end{eqnarray}
		for each $u$ from the domain of $H$. Since the operators at hand leave the domain $D_H$ invariant, we have that for every $t\geq0$ and for every $u\in D_H$:
		\begin{eqnarray}
			\label{H-12bintegrand}\begin{aligned}
				&\left(H\left(t+H\right)^{-1}\left(H^{-1/2}bH^{-1/2}\right)-\left(H^{-1/2}bH^{-1/2}\right)H\left(t+H\right)^{-1}\right)u=\\
				&\left(t+H\right)^{-1}\left(H\left(H^{-1/2}bH^{-1/2}\right)\left(t+H\right)-\left(t+H\right)\left(H^{-1/2}bH^{-1/2}\right)H\right)\left(t+H\right)^{-1}u=\\
				&t\left(t+H\right)^{-1}\left(H\left(H^{-1/2}bH^{-1/2}\right)-\left(H^{-1/2}bH^{-1/2}\right)H\right)\left(t+H\right)^{-1}u=\\
				&t\left(t+H\right)^{-1}\left(H^{-1/2}[H,b]H^{-1/2}\right)\left(t+H\right)^{-1}u=\\
				&-\hbar\omega t(t+H)^{-1} \left(H^{-1/2}bH^{-1/2}\right)(t+H)^{-1}u.
			\end{aligned}
		\end{eqnarray}
		Substituting \eqref{H-12bintegrand} into \eqref{H-12bcomp} gives
		\begin{equation}\label{H-12bcompint}
			[H^{-1/2},b]u=\frac{\hbar\omega}{\pi}\int_0^\infty t^{1/2}(t+H)^{-1}\left(H^{-1/2}bH^{-1/2}\right)(t+H)^{-1}udt,
		\end{equation}
		for every $u$ from the domain of $H$. 
		
		On the other hand, notice that the operator Sylvester equation
		\begin{equation}\label{sylh-12}
			\left(-H^{1/2}\right)y+y\left(-H^{1/2}\right)=H^{-1/2}bH^{-1/2}\end{equation}
		is regular on $D_H$, meaning that its solution is unique, bounded, and is given as 
		\begin{equation}\label{y}
			yu=\frac{1}{\pi}\int_0^\infty t^{1/2}\left(t+\left(-H^{1/2}\right)^2\right)^{-1}\left(H^{-1/2}bH^{-1/2}\right)\left(t+\left(-H^{1/2}\right)^2\right)^{-1}udt,
		\end{equation}
		see \cite{BD0}, \cite{BD1}, \cite{BDND1}, \cite{HMMO}, and \cite{VQP}.
		By writing $\bar{b}$ instead of $y$, the formula \eqref{y} applies and extends to the entire Hilbert space $\H$. Therefore, \eqref{H-12bcompint} reduces to 
		$$[H^{-1/2},b]=\hbar\omega \bar{b}.$$
		An analogous procedure goes for showing that
		$$[H^{-1/2},b^\dagger]=-\hbar\omega \bar{b}^\dagger,$$
		where $\bar{b}^\dagger=\left(\bar{b}\right)^\dagger$.\end{proof}
	
	Since $a$ and $a^\dagger$ do not commute, it is natural to suspect that $T_0$ and $T_0^\dagger$ do not commute either. In fact, since $[T_0,T_0^\dagger]$ must be bounded and self-adjoint, we have the following result: 
	\begin{theorem}
		The operators $T_0$ and $T_0^\dagger$ from \eqref{TTdag} satisfy:
		\begin{equation}
			\label{T0T0dagcomm}
			[T_0,T_0^\dagger]= \hbar\omega H^{-1} -\frac{\hbar\omega}{2}\Re\left(H^{-1}bH^{-1}b^\dagger+(b^\dagger \bar{b}+b\bar{b}^\dagger)H^{-1/2}\right).
		\end{equation}
		
	\end{theorem}
	\begin{proof}
		First, direct computations give
		\begin{eqnarray}\label{T0commcomp}
			\begin{aligned}
				&4[T_,T_0^\dagger]=H^{-1/2}(bb^\dagger-b^\dagger b) H^{-1/2}+\left(bH^{-1}b^\dagger-b^\dagger H^{-1}b\right)+\\
				&[bH^{-1/2},b^\dagger H^{-1/2}]+[H^{-1/2}b,H^{-1/2}b^\dagger].
			\end{aligned}
		\end{eqnarray}
		The first summand in \eqref{T0commcomp} is  evaluated by \eqref{Hb} as
		\begin{equation}\label{firstsum}
			H^{-1/2}(bb^\dagger-b^\dagger b)H^{-1/2}=\hbar\omega H^{-1}.
		\end{equation}
		Next, by writing $bH^{-1}$ as 
		$$bH^{-1}=H^{-1}b-[H^{-1},b]=H^{-1}b-\hbar\omega H^{-1}bH^{-1},$$
		where the latter is obtained by utilizing \eqref{H-1b}, we rewrite the second summand in \eqref{T0commcomp} as
		$$\begin{aligned}&bH^{-1}b^\dagger-b^\dagger H^{-1}b=\\
			&H^{-1}bb^\dagger-\hbar\omega H^{-1}bH^{-1}b^\dagger-b^\dagger H^{-1}b=\\
			&[H^{-1}b,b^\dagger]-\hbar\omega H^{-1}bH^{-1}b^\dagger.
		\end{aligned}$$
		By invoking that
		\begin{equation}\label{ACBcomm}[AB,C]=A[B,C]+[A,C]B,
		\end{equation}
		the latter transforms into
		$$[H^{-1}b,b^\dagger]-\hbar\omega H^{-1}bH^{-1}b^\dagger=H^{-1}[b,b^\dagger]+[H^{-1},b^\dagger]b-\hbar\omega H^{-1}bH^{-1}b^\dagger.$$
		The first two summands of the right-hand-side in the last equation are evaluated by \eqref{Hb} and \eqref{H-1b}, respectively, and combined give
		\begin{eqnarray}\label{secondsumalmost}
			\begin{aligned}
				&bH^{-1}b^\dagger-b^\dagger H^{-1}b=\\
				&\hbar\omega H^{-1}-\hbar\omega H^{-1}b^\dagger H^{-1}b-\hbar\omega H^{-1}bH^{-1}b^\dagger=\\
				&\hbar\omega H^{-1}-2\hbar\omega H^{-1}\Re\left(bH^{-1}b^\dagger\right).
			\end{aligned}
		\end{eqnarray}
		Note that $H^{-1}$ and $bH^{-1}b^\dagger$ $*-$commute: this follows from
		$$\begin{aligned}
			&[H^{-1},bH^{-1}b^\dagger]=H^{-1}bH^{-1}b^\dagger-bH^{-1}b^\dagger H^{-1}=\\
			&\frac{1}{\hbar\omega}\left([H^{-1},b]b^\dagger+b[H^{-1},b^\dagger]\right)=\\
			&\frac{1}{\hbar\omega}\left(H^{-1}bb^\dagger -bH^{-1}b^\dagger +bH^{-1}b^\dagger-bb^\dagger H^{-1}\right)=\\
			&\frac{1}{\hbar\omega}\left(H^{-1}bb^\dagger-bb^\dagger H^{-1}\right)=0,
		\end{aligned}$$
		where the last equality follows from the fact that 
		$$bb^\dagger=\hbar\omega aa^\dagger=H+\frac{\hbar\omega}{2}I,$$
		i.e., $bb^\dagger$ and $H^{-1}$ commute. Dually, $b^\dagger b$ and $H^{-1}$ commute as well, and consequently $H^{-1}$ and $b^\dagger H^{-1}b$ commute. This proves that 
		$$[H^{-1},\Re\left(bH^{-1}b^\dagger\right)]=0.$$
		The two elements $H^{-1}$ and $\Re\left(bH^{-1}b^\dagger\right)$ are bounded, self-adjoint, and commutative, therefore their product is also self-adjoint. In other words, 
		$$H^{-1}\Re\left(bH^{-1}b^\dagger\right)=\Re\left(H^{-1}bH^{-1}b^\dagger\right).$$
		Returning to \eqref{secondsumalmost}, this gives
		\begin{equation}\label{secondsum}
			bH^{-1}b^\dagger-b^\dagger H^{-1}b=
			\hbar\omega H^{-1}-2\hbar\omega \Re\left(H^{-1}bH^{-1}b^\dagger\right).
		\end{equation}
		Finally, the summands $[bH^{-1/2},b^\dagger H^{-1/2}]$ and $[H^{-1/2}b,H^{-1/2}b^\dagger]$ are evaluated by first utilizing the algebraic identity \eqref{ACBcomm} to obtain
		\begin{eqnarray}\label{thirdsum}\begin{aligned}
				&[bH^{-1/2},b^\dagger H^{-1/2}]=b[H^{-1/2},b^\dagger H^{-1/2}]+[b,b^\dagger H^{-1/2}]H^{-1/2}=\\
				&b[H^{-1/2},b^\dagger]H^{-1/2}+[b,b^\dagger]H^{-1}+b^\dagger[b,H^{-1/2}]H^{-1/2}=\\
				&-\hbar\omega b\bar{b}^\dagger H^{-1/2}+\hbar\omega H^{-1}-\hbar\omega b^\dagger \bar{b}H^{-1/2}=\\
				&\hbar\omega H^{-1}-\hbar\omega\left(b\bar{b}^\dagger+b^\dagger \bar{b}\right)H^{-1/2},
		\end{aligned}\end{eqnarray}
		and, analogously,
		\begin{equation}\label{fourthsum}[H^{-1/2}b,H^{-1/2}b^\dagger]=\hbar\omega H^{-1}-\hbar\omega H^{-1/2}\left(\bar{b}b^\dagger+\bar{b}^\dagger b\right).\end{equation}
		Substituting \eqref{firstsum}, \eqref{secondsum}, \eqref{thirdsum}, and \eqref{fourthsum} into \eqref{T0commcomp} gives
		$$[T_0,T_0^\dagger]= \hbar\omega H^{-1} -\frac{\hbar\omega}{2}\Re\left(H^{-1}bH^{-1}b^\dagger+(b^\dagger \bar{b}+b\bar{b}^\dagger)H^{-1/2}\right),$$
		that is, the formula \eqref{T0T0dagcomm} is true. \end{proof}
	Now having proved some fundamental properties for the operators $T_0$ and $T_0^\dagger$, it is clear that the operators 
	\begin{equation}\label{C0} C_0:=\frac{1}{2}\left(T_0+T_0^\dagger\right)\end{equation} and
	\begin{equation}\label{S0}
		S_0=\frac{1}{2i}\left( T_0- T_0^\dagger\right)
	\end{equation}
	are bounded and self-adjoint. Naturally we call the operator $C_0$ the cosine operator while we refer to the operator $S_0$ as the sine operator. 
	\begin{theorem} The operators $C_0$ and $S_0$ from \eqref{C0} and \eqref{S0} satisfy:
		\begin{equation}\label{HS0C0}
			[H,S_0]=i\hbar\omega C_0,\quad [H,C_0]=-i\hbar\omega S_0,\quad [C_0,S_0]=-\frac{i}{2}[T_0,T_0^\dagger].
		\end{equation}
		Consequently, all commutators in \eqref{HS0C0} are skew-symmetric. \end{theorem}
	
	\begin{proof}
		Directly.
	\end{proof}
	
	In order to get a better sense of how the Sine and Cosine operators $S_0$ and $C_0$ act on Hermite polynomials, we proceed to study the opretaros $T_0$ and $T_0^\dagger$ on the Hermite polynomials $\Psi_n$. By combining \eqref{an} and \eqref{H-1/2n}, we conclude the following:
	
	\begin{eqnarray}\begin{aligned}\label{T0nexact}
			T_0|n\rangle&=\frac{\sqrt{\hbar \omega}}{2}\left(aH^{-1/2}+H^{-1/2}a\right)|n\rangle=\\
			&=\begin{cases}0,\quad n=0,\\\frac{\sqrt{n}}{2}\left(\frac{1}{\sqrt{n+1/2}}+\frac{1}{\sqrt{n-1/2}}\right)|n-1\rangle,
				\quad n\geq1.
			\end{cases}	
	\end{aligned}\end{eqnarray} 
	Conversely, the action of the operator $T_0^\dagger$ on the Hermite polynomials $\Psi_n$ is calculated as
	\begin{eqnarray}\begin{aligned}\label{T0dagnexact}
			T_0^\dagger|n\rangle&=\frac{\sqrt{\hbar \omega}}{2}\left(a^\dagger H^{-1/2}+H^{-1/2}a^\dagger\right)|n\rangle=\\
			&=\frac{\sqrt{n+1}}{2}\left(\frac{1}{\sqrt{n+1/2}}+\frac{1}{\sqrt{n+3/2}}\right)|n+1\rangle,
			\quad n\in\NN_0.		
	\end{aligned}\end{eqnarray}
	It immediately follows that
	\begin{eqnarray}\begin{aligned}\label{TdagTnexact}
			T_0^\dagger T_0|n\rangle&=\frac{\hbar \omega}{4}\left(a^\dagger H^{-1/2}+H^{-1/2}a^\dagger\right)\left(aH^{-1/2}+H^{-1/2}a\right)|n\rangle=\\
			&=\begin{cases}0,\quad n=0,\\\frac{n}{4}\left(\frac{1}{\sqrt{n+1/2}}+\frac{1}{\sqrt{n-1/2}}\right)^2|n\rangle, \quad n\geq1,
			\end{cases}	
	\end{aligned}\end{eqnarray}
	while
	\begin{eqnarray}\begin{aligned}\label{TTdagnexact}
			T_0T_0^\dagger|n\rangle&=\frac{\hbar \omega}{4}\left(aH^{-1/2}+H^{-1/2}a\right)\left(a^\dagger H^{-1/2}+H^{-1/2}a^\dagger\right)|n\rangle=\\
			&=\frac{n+1}{4}\left(\frac{1}{\sqrt{n+1/2}}+\frac{1}{\sqrt{n+3/2}}\right)^2|n\rangle,
			\quad n\in\NN_0.
	\end{aligned}\end{eqnarray}
	Now we compute
	\begin{equation}\label{C0nexact} C_0|n\rangle=\begin{cases}\frac{\sqrt{2}}{4}\left(1+\frac{1}{\sqrt{3}}\right)|1\rangle,\quad n=0,\\\frac{1}{4}\left(\left(\frac{1}{\sqrt{1-\frac{1}{2n}}}+\frac{1}{\sqrt{1+\frac{1}{2n}}}\right)|n-1\rangle+\left(\frac{1}{\sqrt{1-\frac{1}{2(n+1)}}}+\frac{1}{\sqrt{1+\frac{1}{2(n+1)}}}\right)\right)|n+1\rangle,\quad n\in\NN,
		\end{cases}
	\end{equation}
	and
	\begin{equation}
		\label{S0nexact} S_0|n\rangle=\begin{cases}\frac{i\sqrt{2}}{4}\left(1+\frac{1}{\sqrt{3}}\right)|1\rangle,\quad n=0,\\\frac{1}{4i}\left(\left(\frac{1}{\sqrt{1-\frac{1}{2n}}}+\frac{1}{\sqrt{1+\frac{1}{2n}}}\right)|n-1\rangle-\left(\frac{1}{\sqrt{1-\frac{1}{2(n+1)}}}+\frac{1}{\sqrt{1+\frac{1}{2(n+1)}}}\right)\right)|n+1\rangle,\quad n\in\NN,
		\end{cases}
	\end{equation}
	
	\subsection{Asymptotic analysis and eigenfunctions}
	For a sufficiently large $n$, the expressions $(n\pm1/2)^{-1/2}$ transform into
	$$(n\pm1/2)^{-1/2}=\frac{1}{\sqrt{n}}\left(1\pm\frac{1}{2n}\right)^{-1/2}\approx \frac{1}{\sqrt{n}}\left(1\mp\frac{1}{4n}\right),\quad n>>1,$$
	and $T_0$ is asymptotically estimated as
	\begin{equation}
		\label{Tn} T_0|n\rangle\approx |n-1\rangle,\quad n>>1,\end{equation}
	i.e., $T_0$ approximately acts as the lowering operator. 
	By writing $n+1/2=n+1-1/2$ and $n+3/2=n+1+1/2$ we get the same estimate for $T_0^\dagger$, i.e., 
	$$(n+1\pm1/2)^{-1/2}=\frac{1}{\sqrt{n+1}}\left(1\pm\frac{1}{2(n+1)}\right)^{-1/2}\approx \frac{1}{\sqrt{n+1}}\left(1\mp\frac{1}{4(n+1)}\right),$$
	when $n>>1$, therefore $T_0^\dagger$ is asymptotically estimated as
	\begin{equation}\label{Tdagn}T_0^\dagger|n\rangle\approx |n+1\rangle,\quad n>>1,\end{equation}
	i.e., $T_0^\dagger$ is approximately the raising operator. 
	Finally, this implies that
	
	\begin{equation}
		\label{C0nS0n} C_0|n\rangle\approx \frac{1}{2}\left(|n-1\rangle+|n+1\rangle\right),\quad S_0|n\rangle\approx\frac{1}{2i}\left(|n-1\rangle-|n+1\rangle\right),\quad n>>1.
	\end{equation}

	This shows that the operators $T_0$, $T_0^\dagger$, $C_0$, and $S_0$ asymptotically behave (mathematically speaking) as the operators introduced by Susskind  and Glogower in \cite{LSJG}: indeed, the operators
    \begin{equation}\label{TSG} T_{SG}:= (N+1)^{-1/2} a,\quad T^\dagger_{SG}:=a^\dagger(N+1)^{-1/2},\end{equation}
    introduced by Susskind and Glogower in the said paper \cite{LSJG}, behave as the lowering and raising operators:
\begin{equation}\label{TSGn} T_{SG}|n+1\rangle=n,\quad T_{SG}|0\rangle=0,\quad T^\dagger_{SG}|n\rangle=n+1,\quad n\in\NN_0.\end{equation}
Therefore, the cosine and sine operators, defined in \cite{LSJG} in the natural way:    \begin{equation}\label{SGsinecos}C_{SG}:=\frac{1}{2}\left(T_{SG}+T^\dagger_{SG}\right),\quad S_{SG}:=\frac{1}{2i}\left(T_{SG}-T^\dagger_{SG}\right),
    \end{equation}
satisfy the commutator property:
\begin{equation}\label{CSGSSGcomm}
[C_{SG},S_{SG}]|0\rangle =1;\quad [C_{SG},S_{SG}]|n\rangle=0,\ n\in\NN,
\end{equation}
where this rank-one commutator indicates that $C_{SG}$ and $S_{SG}$ cannot be observed simultaneously (see \cite{LSJG} for details), justifying the time-energy uncertianty principle (see below).\\

Given that our operators $T_0$, $T^\dagger_0$, $C_0$, and $S_0$ asymptotically do behave as the respective operators $T_{SG}$, $T^\dagger_{SG}$, $C_{SG}$, and $S_{SG}$ from \cite{LSJG}, in what follows, we proceed to estimate by how much the operators $T_0$, $T_0^\dagger$, $C_0$, and $S_0$ deviate from their asymptotic approximations \eqref{Tn}--\eqref{C0nS0n}. We conduct this analysis by introducing a perturbation operator $M$, which is $H$-bounded and defined in a supset of $D_H$, i.e., $D_M\supset D_H$. The idea is to introduce a perturbed time operator $T_M$ and its adjoint $T_M^\dagger$, provided in the manner that
	
	\begin{equation}
		\label{TM} T_M=\frac{1}{2}\left(H^{-1/2}b+bH^{-1/2}\right)+M,\quad T_M^\dagger=\frac{1}{2}\left(H^{-1/2}b^\dagger+b^\dagger H^{-1/2}\right)+M^\dagger,
	\end{equation}
	
	for the said $H-$bounded operator $M$. As before, one has
	\begin{equation}\label{CM}
		C_M:=\frac{1}{2}\left(T_M+T_M^\dagger\right),\quad\end{equation} \begin{equation}\label{SM}S_M:=\frac{1}{2i}\left(T_M- T_M^\dagger\right).
	\end{equation}
	Note that, due to the time-energy uncertainty principle, the time operators $T_M$ and $T_M^\dagger$, and consequently, the operators $C_M$ and $S_M$, are not uniquely determined at the ground state $|0\rangle$. This is because at the lowest state, the phase $\phi_0$ is not unique but rather experiences arbitrary dispersion over $[0,2\pi)$, indicating that the energy fluctuates even when there are no particles occupying the ground state. As such, the phase $\phi_0$ is often regarded as a quasi-probability distribution. We treat this ambiguity by introducing some arbitrary vectors $v(\phi_0)$ and $u(\phi_0)\in\H$, $\|u\|$, $\|v\|\leq 1$, and we point out that $\phi_0$ cannot be recovered from them (see \cite{LSJG}). Consequently, for fixed (but arbitrary) vectors $v(\phi_0)$ and $u(\phi_0)$, we define
	\begin{equation}\label{c0s0}
		c_0:=\frac{1}{2}\left(v(\phi_0)+u(\phi_0)\right),\quad s_0:=\frac{1}{2i}\left(v(\phi_0)-u(\phi_0)\right).	
		\end{equation}
    In \cite{LSJG} it was assumed that $c_0=\frac{1}{2}|1\rangle$, and $s_0=\frac{i}{2}|1\rangle$, which need not be the only possibility.
	Below we proceed to evaluate admissible operator(s) $M$ for which the cosine operator $C_M$ and the sine operator $S_M$ satisfy the sought properties:
	\\
	\begin{equation}\label{CMn} C_M|n\rangle=\frac{1}{2}\left(|n-1\rangle+|n+1\rangle\right),\quad n\geq1; \quad C_M|0\rangle=c_0.
	\end{equation}
	\begin{equation}\label{SMn}
		S_M|n\rangle=\frac{1}{2i}\left(|n-1\rangle-|n+1\rangle\right),\quad n\geq1;\quad S_M|0\rangle=s_0.
	\end{equation}
	Define the operator $L$ as:
	\begin{equation}
		\label{L} L|n\rangle=\begin{cases} 0,\quad n=0,\\\left(1-\frac{\sqrt{n}}{2}\left(\frac{1}{\sqrt{n-1/2}}+\frac{1}{\sqrt{n+1/2}}\right)\right)|n\rangle,\quad n\in\NN.
		\end{cases}
	\end{equation}
	\begin{lemma} The operator $L$ is self-adjoint and in the trace-class. 
	\end{lemma}
	\begin{proof} That $L$ is self-adjoint follows from the SVEP and its definition in \eqref{L}. Moreover, by analyzing the real positive sequence
		\begin{equation}
			\label{lambda}
			\lambda_n=1-\frac{\sqrt{n}}{2}\left(\frac{1}{\sqrt{n-1/2}}+\frac{1}{\sqrt{n+1/2}}\right)
		\end{equation}
		we see that $\lambda_n$ is decreasing and tends to zero when $n\to\infty$. By the same estimate for Taylor expansion,
		$$(n\pm1/2)^{-1/2}=\frac{1}{\sqrt{n}}\left(1\pm\frac{1}{2n}\right)^{-1/2}= \frac{1}{\sqrt{n}}\left(1\mp\frac{1}{4n}+O\left(\frac{1}{n^2}\right)\right),\quad n\in\NN,$$
		we have that
		\begin{equation}
			\label{alphaasym}
			\lambda_n=O\left(\frac{1}{n^2}\right),\quad n\to\infty.
		\end{equation}
		This proves that the sequence $\left(\lambda_n\right)_{n\in\NN}$ is in $\ell_1(\RR)$. Equivalently, $L$ is a self-adjoint trace-class operator.
	\end{proof}
	Let $V$ be defined as
	\begin{equation}\label{Vn} V|n\rangle=|n-1\rangle, \ \quad n\in\NN,\quad V|0\rangle=0.
	\end{equation}
	Accordingly, its conjugate $V^\dagger$ is the opposite shift
	\begin{equation}
		\label{Vdagn}V^\dagger|n\rangle=|n+1\rangle, n\in\NN_0.
	\end{equation}
	Note that $V^\dagger$ is the unilateral shift.

	\begin{theorem} For a given $H-$bounded linear operator $M$, let $T_M$, $T^\dagger_M$, $C_M$, and $S_M$ be defined as \eqref{TM}, \eqref{CM}, and \eqref{SM}, respectively. The following claims are equivalent:
		\begin{itemize}
			\item[(a)] The operator $C_M$ satisfies \eqref{CMn} and the operator $S_M$ satisfies \eqref{SMn}.
			\item[(b)] The operator $M$ is a (bounded) trace-class operator, and is exactly of the form 
			\begin{equation}\label{Mn}M|n\rangle=\begin{cases}v(\phi_0),\quad n=0;\\
					VL|n\rangle,\quad n\in\NN,
				\end{cases}
			\end{equation}
			where $L$ is defined as \eqref{L}, and $V$ is the lowering operator \eqref{Vn}.
		\end{itemize}
	\end{theorem}
	\begin{proof} We start with the lowest eigenstate $|0\rangle$. By \eqref{T0nexact} and \eqref{T0dagnexact} we have
		\begin{equation}\label{SM0comp}
			S_M|0\rangle=s_0\Leftrightarrow \frac{1}{2i}\left(M-M^\dagger\right)|0\rangle-s_0=\frac{1}{4i}\left(\frac{1}{\sqrt{1/2}}+\frac{1}{\sqrt{3/2}}\right)|1\rangle.
		\end{equation}
		On the other hand, we have
		\begin{eqnarray}\label{CM0comp}
			\begin{aligned}
				c_0=&C_M|0\rangle\Leftrightarrow \\
				c_0=&\frac{1}{2}\left(M+M^\dagger\right)|0\rangle+\frac{1}{4}\left(\frac{1}{\sqrt{1/2}}+\frac{1}{\sqrt{3/2}}\right)|1\rangle=M|0\rangle-is_0,
			\end{aligned}
		\end{eqnarray}
		Giving 
		\begin{equation}\label{M0Mdag0}
			M|0
			\rangle=c_0+is_0,\quad M^\dagger|0\rangle=c_0-is_0.
		\end{equation}

		For $n>0$, it follows from \eqref{T0nexact}, \eqref{T0dagnexact}, \eqref{TM}, and \eqref{CM} that:
		\begin{eqnarray}\label{CMncomp}
			\begin{aligned}C_M|n\rangle=&\frac{1}{2}\left(M+M^\dagger\right)|n\rangle+\\
				+&\frac{1}{2}\left(\frac{\sqrt{n}}{2}\left(\frac{1}{\sqrt{n-1/2}}+\frac{1}{\sqrt{n+1/2}}\right)\right)|n-1\rangle+\\
				+&\frac{1}{2}\left(\frac{\sqrt{n+1}}{2}\left(\frac{1}{\sqrt{n+1/2}}+\frac{1}{\sqrt{n+3/2}}\right)\right)|n+1\rangle.
			\end{aligned}
		\end{eqnarray}
		Therefore \eqref{CMn} holds if and only if
		\begin{eqnarray}\label{MCncomp}
			\begin{aligned}
				&\left(M+M^\dagger\right)|n\rangle=\\
				&\left(1-\frac{\sqrt{n}}{2}\left(\frac{1}{\sqrt{n-1/2}}+\frac{1}{\sqrt{n+1/2}}\right)\right)|n-1\rangle+\\
				+&\left(1-\frac{\sqrt{n+1}}{2}\left(\frac{1}{\sqrt{n+1/2}}+\frac{1}{\sqrt{n+3/2}}\right)\right)|n+1\rangle=\\
				=&\left(\lambda_n|n-1\rangle+\lambda_{n+1}|n+1\rangle\right)=\\
				=& \left(VL|n\rangle+LV^\dagger|n\rangle\right)=\\
				=&\left(VL+LV^\dagger\right)|n\rangle.
			\end{aligned}
		\end{eqnarray}
		Since $L$ is trace-class and self-adjoint, while $V$ is a partial isometry, it follows that $VL$ is bounded and $(VL)^\dagger=LV^\dagger$. Then \eqref{MCncomp} transforms into
		\begin{equation}\label{Mnreduced}
			\left(M+M^\dagger\right)|n\rangle=\left(VL+\left(VL\right)^\dagger\right)|n\rangle, \quad n\in\NN.
		\end{equation}
		If $M=VL$ then $M^\dagger=(VL)^\dagger$ and vice-versa. In this case $M$ and $M^\dagger$ are trace-class. Otherwise, if $M=\left(VL+F\right)$, where $F\neq0$, then $F$ is also $H-$bounded with the same domain as $M$, and in that case
		$M^\dagger=\left(VL\right)^\dagger+F^\dagger$. But then
		$$\left(M+M^\dagger\right)|n\rangle=\left(VL+\left(VL\right)^\dagger+F+F^\dagger\right)|n\rangle=\left(VL+\left(VL\right)^\dagger\right)|n\rangle$$
		implying that
		$$\left(F+F^\dagger\right)|n\rangle=0,\quad n\in\NN,$$
		so $F$ is skew-symmetric.

		On the other hand, from \eqref{T0nexact}, \eqref{T0dagnexact}, \eqref{TM}, and \eqref{SM} it follows that for every $n\in\NN$:
		\begin{eqnarray}\label{SMncomp}
			\begin{aligned}
				S_M|n\rangle=&\frac{1}{2}\left(M-M^\dagger\right)|n\rangle+\\
				+&\frac{1}{2}\left(\frac{\sqrt{n}}{2}\left(\frac{1}{\sqrt{n-1/2}}+\frac{1}{\sqrt{n+1/2}}\right)\right)|n-1\rangle-\\
				+&\frac{1}{2}\left(\frac{\sqrt{n+1}}{2}\left(\frac{1}{\sqrt{n+1/2}}+\frac{1}{\sqrt{n+3/2}}\right)\right)|n+1\rangle.
			\end{aligned}
		\end{eqnarray}
		Therefore \eqref{SMn} holds if and only if
		\begin{eqnarray}\label{MSncomp}
			\begin{aligned}
				&\left(M-M^\dagger\right)|n\rangle=\\
				&\left(1-\frac{\sqrt{n}}{2}\left(\frac{1}{\sqrt{n-1/2}}+\frac{1}{\sqrt{n+1/2}}\right)\right)|n-1\rangle-\\
				-&\left(1-\frac{\sqrt{n+1}}{2}\left(\frac{1}{\sqrt{n+1/2}}+\frac{1}{\sqrt{n+3/2}}\right)\right)|n+1\rangle=\\
				=&\left(\lambda_n|n-1\rangle-\lambda_{n+1}|n+1\rangle\right)=\\
				=& \left(VL|n\rangle-LV^\dagger|n\rangle\right)=\\
				=&\left(VL-\left(VL\right)^\dagger\right)|n\rangle.
			\end{aligned}
		\end{eqnarray}
		As before, if $M=VL$ then $M^\dagger=(VL)^\dagger$. In this case $M$ and $M^\dagger$ are trace-class. Otherwise, if $M=VL+P$, where $P\neq0$, then $P$ is also $H-$bounded with the same domain as $M$, and in that case
		$M^\dagger=\left(\left(VL\right)^\dagger-P^\dagger\right)$. But then
		$$\left(M-M^\dagger\right)|n\rangle=\left(VL-\left(VL\right)^\dagger+P-P^\dagger\right)|n\rangle=\left(VL-\left(VL\right)^\dagger\right)|n\rangle$$
		implying that
		$$\left(P-P^\dagger\right)|n\rangle=0,\quad n\in\NN,$$
		so $P$ is a symmetric densely defined operator. Combining the calculations from \eqref{Mnreduced} and \eqref{MSncomp}, we conclude that either $$VL+P-F=VL,$$
		on $D_M$, or 
		$$VL+P-F=LV^\dagger$$
		for some symmetric $P$ and skew-symmetric $F$, both defined in $D_M$. From the former one has that $P-F=0$, i.e., $P=F=0$, therefore \eqref{Mn} is true. On the other hand, if the latter is true, then one applies the involution $\dagger$ and obtains
		$$VL=LV^\dagger+P+F\Leftrightarrow LV^\dagger=VL-F-P,$$
		implying that $P=0$ and consequently $M|n\rangle=VL|n\rangle$ for $n\in\NN$.
	\end{proof}
	\begin{corollary}
		The operators $S_M$ and $C_M$ are bounded, self-adjoint, and differ from $S_0$ and $C_0$, respectively, by trace-class operators. Moreover,
		\begin{equation}\label{CommCMSM}
			[C_M,S_M]=\frac{1}{2i}\left([T_0,T_0^\dagger]+\left(M^\dagger T_0-T_0M^\dagger\right)+\left( M^\dagger T_0-T_0M^\dagger\right)^\dagger\right).
		\end{equation}
		Consequently, the commutator $[C_M,S_M]$ is skew-self-adjoint.
	\end{corollary}
	The fact that $C_M$ and $S_M$ do not commute is essential in QHO because it indicates that the vectors $v(\phi_0)$ and $u(\phi_0)$ cannot be used to recover the phase $\phi_0$, see \cite{LSJG}. 
	
	\begin{corollary} The operators $C_0$ and $S_0$ are trace-class perturbations of $C_M$ and $S_M$, respectively.
		\end{corollary}

    \begin{corollary}
    The commutator $[T_0,T_0^\dagger]$ is a trace-class operator. Consequently, the commutator $[C_0,S_0]$ is also a trace-class operator, and 
    $H^{-1}$ differs from ${2}\Re\left(H^{-1}bH^{-1}b^\dagger+(b^\dagger \bar{b}+b\bar{b}^\dagger)H^{-1/2}\right)$ by a trace-class operator.
    \end{corollary}
\begin{proof} From \eqref{CSGSSGcomm} it follows that $[C_M,S_M]$ is a rank-one operator. Then, from formula \eqref{CommCMSM} it follows that $[T_0,T_0^\dagger]$ is a trace-class perturbation of the rank-one operator, i.e., it is a trace-class operator. The rest follows from \eqref{T0T0dagcomm} and from \eqref{HS0C0}.
    \end{proof}
	For more on spectral properties of trace-class perturbations of self-adjoint operators, see e.g. \cite{KS}. 
	
	\section{Difference of two phases} \label{two phases} 
	
	For the standard setup (see \cite{NFM91} and \cite{NFM92}), let us have two monochromatic lasers of the same frequency pointed at a 50:50 beam-splitter. Let the amplitudes of the lasers at the half-silvered mirror (the signals) be $E_1(t) = A_1 e^{i \varphi_1 + i\lambda t}$ and $E_2(t) = A_2 e^{i \varphi_2 + i\lambda t}$. 
	Then the outgoing lasers would have amplitudes at the detectors 
	$E_3(t) = \frac{1}{\sqrt{2}} (E_1(t) + i E_2(t))$ and 
	$E_4(t) = \frac{1}{\sqrt{2}} (i E_1(t) + E_2(t))$. The intensities (or the energies) would be: \\ 
	$I_1 \ = \ |E_1(t)|^2 \ = \ A_1^2$ \\ 
	$I_2 \ = \ |E_2(t)|^2 \ = \ A_2^2$ \\ 
	$I_3 \ = \ |E_3(t)|^2 \ = \ \frac{A_1^2 + A_2^2}{2} + A_1 A_2 \sin(\phi_1 - \phi_2)$ \\ 
	$I_4 \ = \ |E_4(t)|^2 \ = \ \frac{A_1^2 + A_2^2}{2} - A_1 A_2 \sin(\phi_1 - \phi_2)$
	\par 
	Therefore 
	$$I_3 - I_4 \ = \ 2A_1 A_2 \sin(\phi_1 - \phi_2)$$ 
	or equivalently 
	\begin{equation} \label{sin} 
		\sin(\phi_1 - \phi_2) \ = \ \frac{I_3 - I_4}{2A_1 A_2} 
	\end{equation} 
	
	Next, by adding a $\lambda/4$ shifter to the second laser, we end up with a modified 
	$$E_2'(t) \ = \ A_2 e^{i(\phi_2 + \frac{\pi}{2} + \lambda t)} \ = \ i A_2 e^{i\phi_2 + i\lambda t} \ = \ i E_2(t)$$ 
	This yields 
	$$E_3'(t) = \frac{1}{\sqrt{2}} (E_1(t) - E_2(t)) \ \ \ E_4'(t) = \frac{1}{\sqrt{2}} (i E_1(t) + i E_2(t))$$
	And thus 
	$$I_3' \ = \ |E_3'(t)|^2 \ = \ \frac{A_1^2+A_2^2}{2} - A_1A_2\cos(\phi_1 - \phi_2)$$
	$$I_4' \ = \ |E_4'(t)|^2 \ = \ \frac{A_1^2+A_2^2}{2} + A_1A_2\cos(\phi_1 - \phi_2)$$
	Thus 
	$$I_3' - I_4' \ = \ - 2 A_1 A_2 \cos(\phi_1 - \phi_2)$$
	or equivalently 
	\begin{equation} \label{cos} 
		\cos(\phi_1 - \phi_2) \ = \ - \frac{I_3' - I_4'}{2 A_1 A_2} 
	\end{equation} 
	
	When attempting to quantize this, we argue that since all $I_i$ and $I_i'$ represent energies, we can take 
	$I_1 \sim H_1$ and $I_2 \sim H_2$, where $H_1$ and $H_2$ are the Hamiltonians of two identical Harmonic oscillators, acting on two different isomorphic copies $\mathcal{H}_1$ and $\mathcal{H}_2$ of the standard Harmonic oscillator Hilbert space.  
	\par
	Let us promote $E_1$ and $E_2$ to operators representing the quantized version of the above classical experiment and denote them $\hat{E_1}$ and $\hat{E_2}$, respectively. From the above equations, we see that $|\hat{E_1}|^2$, $|\hat{E_2}|^2$ behave like $H_1$ and $H_2$, and $|\hat{E_3}|^2$ and $|\hat{E_4}|^2$ $i=3,4$ behave like the resulting Hamiltonians $H_3$ and $H_4$. There is an ambiguity whether to set $\hat{E_i} \hat{E_i}^* = H_i$ or $\hat{E_i}^* \hat{E_i} = H_i$, for $i=1,2,3,4$. This is because, by using the creation and annihilation operators, we can write 
	\begin{equation} \label{EHa} 
		2 |\hat{E}_i|^2 \ = \ a_i^* a_i + a_i a_i^* \ = \ 2H_i, \ \ i = 1,2.
	\end{equation}
	
	However, for the definition of the quantized version $S_{12}$ of $\sin(\phi_1 - \phi_2)$ 
    there is no such ambiguity: 
	$$ 2 \hat{E_3} \hat{E_3}^* \ = \ \hat{E_1} \hat{E_1}^* - i \hat{E_1} \hat{E_2}^* + i \hat{E_2} \hat{E_1}^* + \hat{E_2} \hat{E_2}^* $$
	$$ 2 \hat{E_4} \hat{E_4}^* \ = \ \hat{E_1} \hat{E_1}^* + i \hat{E_1} \hat{E_2}^* - i \hat{E_2} \hat{E_1}^* + \hat{E_2} \hat{E_2}^* $$
	Therefore 
	$$2[\hat{E_3} \hat{E_3}^* - \hat{E_4} \hat{E_4}^*] \ = \ 2 i[\hat{E_2} \hat{E_1}^* - \hat{E_1} \hat{E_2}^*]$$ 
	On the other hand 
	$$ 2 \hat{E_3}^* \hat{E_3} \ = \ \hat{E_1}^* \hat{E_1} + i \hat{E_1}^* \hat{E_2} - i \hat{E_2}^* \hat{E_1} + \hat{E_2}^* \hat{E_2} $$
	$$ 2 \hat{E_4}^* \hat{E_4} \ = \ \hat{E_1}^* \hat{E_1} - i \hat{E_1}^* \hat{E_2} + i \hat{E_2}^* \hat{E_1} + \hat{E_2}^* \hat{E_2} $$
	Therefore, yet again 
	$$2[\hat{E_3}^* \hat{E_3} - \hat{E_4}^* \hat{E_4}] \ = \ 2 i[ \hat{E_1}^* \hat{E_2} -  \hat{E_2}^*\hat{E_1}]$$ 
	Thus, we can attempt to define the $\sin$ operator as 
	
	$$S_{12} \ \ "=" \ \ \frac{\hat{E_3}^* \hat{E_3} - \hat{E_4}^* \hat{E_4}}{2\sqrt{H_1H_2}} \ = \ \frac{\hat{E_3} \hat{E_3}^* - \hat{E_4} \hat{E_4}^*}{2\sqrt{H_1H_2}}$$
	or equivalently 
	\begin{multline} \label{S12} 
		S_{12} \ \ "=" \ \ i\frac{\hat{E_1} \hat{E_2}^* - \hat{E_2} \hat{E_1}^*}{2\sqrt{H_1H_2}} \ = \ i \frac{\hat{E_2}^* \hat{E_1} - \hat{E_1}^* \hat{E_2}}{2\sqrt{H_1H_2}} \\ 
		"="\ \frac{i}{2} \left[ \left(\frac{\hat{E_1}}{\sqrt{\hat{H_1}}}\right) \left(\frac{\hat{E_2}}{\sqrt{\hat{H_2}}}\right)^*  \ - \ 
		\left(\frac{\hat{E_2}}{\sqrt{\hat{H_2}}}\right) \left(\frac{\hat{E_1}}{\sqrt{\hat{H_1}}}\right)^* \right] \\ = \ 
		\frac{i}{2} \left[ \left(\frac{\hat{E_2}}{\sqrt{\hat{H_2}}}\right)^* \left(\frac{\hat{E_1}}{\sqrt{\hat{H_1}}}\right) \ - \ 
		\left(\frac{\hat{E_1}}{\sqrt{\hat{H_1}}}\right)^* \left(\frac{\hat{E_2}}{\sqrt{\hat{H_2}}}\right) \right]
	\end{multline}
	
	As in the $\sin$ case, for the quantized version $C_{12}$ of $\cos(\phi_1 - \phi_2)$, we observe  
    
	$$2\hat{E_3'} \hat{E_3'}^* \ = \ \hat{E_1} \hat{E_1}^* - \hat{E_1} \hat{E_2}^* - \hat{E_2} \hat{E_1}^* + \hat{E_2} \hat{E_2}^*$$
	$$2\hat{E_4'} \hat{E_4'}^* \ = \ \hat{E_1} \hat{E_1}^* + \hat{E_1} \hat{E_2}^* + \hat{E_2} \hat{E_1}^* + \hat{E_2} \hat{E_2}^*$$
	and so 
	$$2[\hat{E_3'} \hat{E_3'}^* - \hat{E_4'} \hat{E_4'}^*] \ = \ -2[\hat{E_1} \hat{E_2}^* + \hat{E_2} \hat{E_1}^*]$$
	On the other hand 
	$$2\hat{E_3'}^* \hat{E_3'} \ = \ \hat{E_1}^* \hat{E_1} - \hat{E_1}^* \hat{E_2} - \hat{E_2}^* \hat{E_1} + \hat{E_2}^* \hat{E_2}$$
	$$2\hat{E_4'}^* \hat{E_4'} \ = \ \hat{E_1}^* \hat{E_1} + \hat{E_1}^* \hat{E_2} + \hat{E_2}^* \hat{E_1} + \hat{E_2}^* \hat{E_2}$$
	and this also gives 
	$$2[\hat{E_3'}^*\hat{E_3'} - \hat{E_4'}^*\hat{E_4'}] \ = \ -2[\hat{E_1}^*\hat{E_2} + \hat{E_2}^*\hat{E_1}]$$
	So we can attempt to define 
	$$C_{12} \ \ "=" \ \ \frac{\hat{E_3'}^* \hat{E_3'} - \hat{E_4'}^* \hat{E_4'}}{2\sqrt{H_1H_2}} \ = \ \frac{\hat{E_3'} \hat{E_3'}^* - \hat{E_4'} \hat{E_4'}^*}{2\sqrt{H_1H_2}}$$
	Or equivalently
	\begin{multline} \label{C12} 
        C_{12} \ "=" \ \frac{E_1E_2^* + E_2E_1^*}{2\sqrt{\hat{H_1} \hat{H_2}}} \ = \ 
        \frac{\hat{E_1}^* \hat{E_2} +   \hat{E_2}^* \hat{E_1}}{2\sqrt{\hat{H_1} \hat{H_2}}} \\
		"=" \ \frac{1}{2} \left[ \left(\frac{\hat{E_1}}{\sqrt{\hat{H_1}}}\right) \left(\frac{\hat{E_2}}{\sqrt{\hat{H_2}}}\right)^* \ + \ 
		\left(\frac{\hat{E_2}}{\sqrt{\hat{H_2}}}\right) \left(\frac{\hat{E_1}}{\sqrt{\hat{H_1}}}\right)^* \right] \\ = \ 
		\frac{1}{2} \left[ \left(\frac{\hat{E_2}}{\sqrt{\hat{H_2}}}\right)^* \left(\frac{\hat{E_1}}{\sqrt{\hat{H_1}}}\right) \ + \ 
		\left(\frac{\hat{E_1}}{\sqrt{\hat{H_1}}}\right)^* \left(\frac{\hat{E_2}}{\sqrt{\hat{H_2}}}\right) \right]
	\end{multline}
	
	It is known that $\sin(\phi_1 - \phi_2)$ and $\cos(\phi_1 - \phi_2)$ can be measured simultaneously classically (\cite{NFM91} and \cite{NFM92}), so it makes sense to assume that $S_{12}$ and $C_{12}$ would ``almost'' commute. Also, since $S_{12}$ and $C_{12}$ are self-adjoint as defined above, we can assume that they would ``almost'' satisfy the Pythagorean identity. So we write the restrictions: 
    \begin{equation} \label{restrictions} 
	\begin{aligned} 
    S_{12} C_{12} - C_{12} S_{12} \ & \sim  \ 0   \\
	S_{12}^2 + C_{12}^2 \ & \sim  \ 4 I  
    \end{aligned} 
    \end{equation}
	Note that 
	$$T_{0i} \text{ is the symmetrized version of } \frac{\hat{E_i}}{\sqrt{\hat{H_i}}}, \ i=1,2,$$
    where $T_{0i} \ = \ \frac{1}{2}(H_i^{-1/2}a_i + a_iH_i^{-1/2})$. Here $H_1$ and $H_2$ are two identical 
    Harmonic oscillator Hamiltonians with annihilation operators $a_1$ and $a_2$, respectively, acting on the states of two isomorphic Hilbert spaces $\mathcal{H}_1$ and $\mathcal{H}_2$. 
    \par
	So we set 
	\begin{equation} \label{S12T} 
		S_{12} \ = \ \frac{i}{2} [T_{01} T_{02}^* - T_{01} T_{02}^*] \ = \ 
		\frac{i}{2} [T_{02}^* T_{01} - T_{01}^* T_{02}] 
	\end{equation}
	\begin{equation} \label{C12T} 
		C_{12} \ = \ \frac{1}{2} [T_{01} T_{02}^* + T_{01} T_{02}^*] \ = \ 
		\frac{1}{2} [T_{02}^* T_{01} + T_{01}^* T_{02}] 
	\end{equation}
	We want the following conditions to be satisfied, in accord to (\ref{restrictions})
	\begin{equation} \label{sim0} 
		T_{01} T_{01}^* T_{02}^* T_{02} \ - \ T_{01}^* T_{01} T_{02} T_{02}^* \ \sim \ 0
	\end{equation}
	\begin{equation} \label{sim2I} 
		T_{01} T_{01}^* T_{02}^* T_{02} \ + \ T_{01}^* T_{01} T_{02} T_{02}^* \ \sim \ 2 I 
	\end{equation}
	It remains to interpret the meaning of $\sim \ 0$ and the meaning of $\sim \ 2 I$. 
	\par
	On an eigenvector $|m\rangle |n\rangle \in \mathcal{H}_1 \otimes \mathcal{H}_2$ of $H_1 \otimes H_2$, we compute 
	
	\begin{eqnarray}\begin{aligned}\label{sim0}
			[T_{0M1} T_{0M1}^* T_{0M2}^* T_{0M2} \ - \ T_{0M1}^* T_{0M1} T_{0M2} T_{0M2}^*]|m \rangle |n\rangle = \\
			=\begin{cases}0,\quad m=n=0,\\ 
				\frac{n}{16}\left(\frac{1}{\sqrt{1/2}}+\frac{1}{\sqrt{3/2}}\right)^2\left(\frac{1}{\sqrt{n+1/2}}+\frac{1}{\sqrt{n-1/2}}\right)^2|0\rangle |n\rangle, \quad m=0, n\geq1,\\
				-\frac{m}{16}\left(\frac{1}{\sqrt{1/2}}+\frac{1}{\sqrt{3/2}}\right)^2\left(\frac{1}{\sqrt{m+1/2}}+\frac{1}{\sqrt{m-1/2}}\right)^2|m\rangle |0\rangle, \quad m\geq1, n=0,\\
				\left[ \frac{(m+1)n}{16} \left(\frac{1}{\sqrt{m+1/2}}+\frac{1}{\sqrt{m+3/2}}\right)^2\left(\frac{1}{\sqrt{n+1/2}}+\frac{1}{\sqrt{n-1/2}}\right)^2 - \right. \\ \left. \frac{m(n+1)}{16}\left(\frac{1}{\sqrt{m+1/2}}+\frac{1}{\sqrt{m-1/2}}\right)^2\left(\frac{1}{\sqrt{n+1/2}}+\frac{1}{\sqrt{n+3/2}}\right)^2 \right] |m\rangle |n\rangle, m\geq1,n\geq1
			\end{cases}	
	\end{aligned}\end{eqnarray}
	
	\begin{eqnarray}\begin{aligned}\label{sim0}
			[T_{0M1} T_{0M1}^* T_{0M2}^* T_{0M2} \ + \ T_{0M1}^* T_{0M1} T_{0M2} T_{0M2}^*]|m \rangle |n\rangle = \\
			=\begin{cases}0,\quad m=n=0,\\ 
				\frac{n}{16}\left(\frac{1}{\sqrt{1/2}}+\frac{1}{\sqrt{3/2}}\right)^2\left(\frac{1}{\sqrt{n+1/2}}+\frac{1}{\sqrt{n-1/2}}\right)^2|0\rangle |n\rangle, \quad m=0, n\geq1,\\
				\frac{m}{16}\left(\frac{1}{\sqrt{1/2}}+\frac{1}{\sqrt{3/2}}\right)^2\left(\frac{1}{\sqrt{m+1/2}}+\frac{1}{\sqrt{m-1/2}}\right)^2|m\rangle |0\rangle, \quad m\geq1, n=0,\\
				\left[ \frac{(m+1)n}{16} \left(\frac{1}{\sqrt{m+1/2}}+\frac{1}{\sqrt{m+3/2}}\right)^2\left(\frac{1}{\sqrt{n+1/2}}+\frac{1}{\sqrt{n-1/2}}\right)^2 + \right. \\ \left. \frac{m(n+1)}{16}\left(\frac{1}{\sqrt{m+1/2}}+\frac{1}{\sqrt{m-1/2}}\right)^2\left(\frac{1}{\sqrt{n+1/2}}+\frac{1}{\sqrt{n+3/2}}\right)^2 \right] |m\rangle |n\rangle, m\geq1,n\geq1
			\end{cases}	
	\end{aligned}\end{eqnarray}
	\par 
	Denote $\mathcal{H}_i' \ = \ \text{span} \{ |n\rangle \ | \ n \geq 1 \}$. 
	\begin{proposition}
		(i) The restriction to $\mathcal{H}_1' \otimes \mathcal{H}_2'$ of the operator 
		$$T_{0M1} T_{0M1}^* T_{0M2}^* T_{0M2} \ - \ T_{0M1}^* T_{0M1} T_{0M2} T_{0M2}^*$$ 
		is equal to $0$ up to a compact operator. 
		\par
		(ii) The restriction on $\mathcal{H}_1' \otimes \mathcal{H}_2'$ of the operator  
		$$T_{0M1} T_{0M1}^* T_{0M2}^* T_{0M2} \ + \ T_{0M1}^* T_{0M1} T_{0M2} T_{0M2}^*$$ 
		is equal to $2 I$ up to a compact operator. 
	\end{proposition}
	\begin{proof}
		Using the formula 
		$$(1+x)^{-\frac{1}{2}} \ = \ 1 - \frac{x}{2} + \frac{3x^2}{8} + O(x^3)$$
		we obtain: 
		$$\frac{m+1}{4m}\left[\frac{1}{\sqrt{m+1/2} + \sqrt{m+3/2}}\right]^2 \ = \ 1 + \frac{3}{16m^2}  + O\left(\frac{1}{m^3}\right)$$ 
		$$\frac{n}{4}\left[\frac{1}{\sqrt{n+1/2}} + \frac{1}{\sqrt{n-1/2}}\right]^2 \ = \ 1 + \frac{3}{16n^2} + O\left(\frac{1}{n^4}\right)$$
		Therefore 
		\begin{multline*} 
			\left[ \frac{(m+1)n}{16} \left(\frac{1}{\sqrt{m+1/2}}+\frac{1}{\sqrt{m+3/2}}\right)^2\left(\frac{1}{\sqrt{n+1/2}}+\frac{1}{\sqrt{n-1/2}}\right)^2 - \right. \\ \left. \frac{m(n+1)}{16}\left(\frac{1}{\sqrt{m+1/2}}+\frac{1}{\sqrt{m-1/2}}\right)^2\left(\frac{1}{\sqrt{n+1/2}}+\frac{1}{\sqrt{n+3/2}}\right)^2 \right] \ = \\ 
			\left[1 + \frac{3}{16m^2} - \frac{3}{8m^3} + O\left(\frac{1}{m^4}\right)\right]\left[1 + \frac{3}{16n^2} + O\left(\frac{1}{n^4}\right)\right] \ - \\ 
			\left[1 + \frac{3}{16m^2} + O\left(\frac{1}{m^4}\right)\right]\left[1 + \frac{3}{16n^2} - \frac{3}{8n^3} + O\left(\frac{1}{n^4}\right)\right] \ = \\ 
			-\frac{3}{8m^{3}} + \frac{3}{8n^{3}} + O\!\left(\frac{1}{m^{4}}\right) + O\!\left(\frac{1}{n^{4}}\right) + O\!\left(\frac{1}{m^{3}n^{2}}\right) + O\!\left(\frac{1}{m^2n^3}\right) 
		\end{multline*}
		hence, (i). 
		\par 
		Similarly, 
		\begin{multline*} 
			\left[ \frac{(m+1)n}{16} \left(\frac{1}{\sqrt{m+1/2}}+\frac{1}{\sqrt{m+3/2}}\right)^2\left(\frac{1}{\sqrt{n+1/2}}+\frac{1}{\sqrt{n-1/2}}\right)^2 + \right. \\ \left. \frac{m(n+1)}{16}\left(\frac{1}{\sqrt{m+1/2}}+\frac{1}{\sqrt{m-1/2}}\right)^2\left(\frac{1}{\sqrt{n+1/2}}+\frac{1}{\sqrt{n+3/2}}\right)^2 \right] \ = \\ 
			\left[1 + \frac{3}{16m^2} - \frac{3}{8m^3} + O\left(\frac{1}{m^4}\right)\right]\left[1 + \frac{3}{16n^2} + O\left(\frac{1}{n^4}\right)\right] \ + \\ 
			\left[1 + \frac{3}{16m^2} + O\left(\frac{1}{m^4}\right)\right]\left[1 + \frac{3}{16n^2} - \frac{3}{8n^3} + O\left(\frac{1}{n^4}\right)\right] \ = \\ 
			2 + \frac{3}{8}\left(\frac{1}{m^{2}}+\frac{1}{n^{2}}\right) - \frac{3}{8}\left(\frac{1}{m^{3}}+\frac{1}{n^{3}}\right) + \frac{9}{128}\frac{1}{m^{2}n^{2}} - \frac{9}{128}\left(\frac{1}{m^{3}n^{2}}+\frac{1}{m^{2}n^{3}}\right) + \\ 
			O\!\left(\frac{1}{m^{4}}\right) + O\!\left(\frac{1}{n^{4}}\right) + O\!\left(\frac{1}{m^{3}n^{3}}\right)
		\end{multline*}
		Hence, (ii). 
	\end{proof}
	
	Note that physically the phase of the zero-photon ground state $|0\rangle$ is not defined. 
	Therefore, it makes sense to restrict our operators to the Hilbert spaces $\mathcal{H}_i', \ i=1,2$.  

	\bigskip
	
	\noindent\textbf{Data availability statement.} Data availability not applicable since no data sets were generated during this research.\\
	
	\noindent\textbf{Declaration of conflict of interest.} The authors declare there is no conflict of interest in publishing the results obtained in this research.\\
	
	\noindent\textbf{Acknowledgements.} The first author is supported by the Ministry of Science, Technological Development and Innovations, Republic of Serbia, grant No. 451-03-66/2026-03/200029, and by the Bulgarian Ministry of Education
	and Science, Scientific Programme ``Enhancing the Research Capacity in Mathematical Sciences (PIKOM)"
	No. DO1-67/05.05.2022.
	\par
	The second author is partially supported by the Bulgarian National Science Fund under Grant No. KP-06-H92/6 (December 08, 2025) and by Grant No 80-10-31/01.04.2026 of the Research Fund of Sofia University. \\


\begin{thebibliography}{00}
		
		\bibitem{BP90} S. Barnett, D. Pegg {\it Quantum Theory of Rotation Angles}. Physical Review A, vol. 41, No 7, pp. 3427-3435, 1990. 
		\bibitem{CN65} P. Carruthers, M.M. Nieto {\it Coherent States and the Number-Phase Uncertainty Relations}. Physical Review Letters, vol. 14, No 11, pp. 387-389, 1965. 
		
		\bibitem{BD0} B. D. Djordjevi\'c, {\it On a singular Sylvester equation with unbounded self-adjoint $A$ and $B$}, Complex Analysis Operator Theroy 14:\textbf{43} (2020) https://doi.org/10.1007/s11785-020-01000-7.
		
		\bibitem{BD1} B. D. Djordjevi\'c, {\it Singular Lyapunov operator equations: application to $C^*-$algebras, Fr\'echet derivatives and abstract Cauchy problems}, Analysis Mathematical Physics (2021) 11:\textbf{160} https://doi.org/10.1007/s13324-021-00596-z
		
		\bibitem{BDND1} B. D. Djordjevi\'c and  N. \v C. Din\v ci\'c, {\it Solving the operator equation $AX-XB=C$ with closed $A$ and $B$}, Integral Equations Operator Theory \textbf{90} (51) (2018)
		https://doi.org/10.1007/s00020-018-2473-3.
		
		
		\bibitem{Kato} T. Kato, {\it Perturbation Theory for Linear Operators}. Grundlehren der Mathematischen Wissenschaften, 2nd ed. (Springer-Verlag, 1976). 
		\bibitem{NFM91} Noh, Fougeres, Mandel, {\it Measurement of the quantum phase by photon counting}. Physical Review Letters, vol. 67(11), pp. 1426-1429, 1991. 
		
		\bibitem{HMMO} H. Mohamad and M. Oliver, {\it $H^S$-class construction of an almost invariant slow subspace for the Klein-Gordon equation in the non-relativistic limit}, Journal of Mathematical Physics \textbf{59}, 051509 (2018)\\
		https://doi:10.1063/1.5027040.
		
		\bibitem{NFM92} Noh, Fougeres, Mandel, {\it Operational Approach to the Phase of a Quantum Field}. Phys. Rev. A, vol. 45, no. 1, pp. 424-443, Jan 1992.
		
		\bibitem{VQP} V. Q. Ph\'{o}ng, \textit{The operator equation $AX-XB=C$ with unbounded operators $A$ and $B$ and related abstract Cauchy problems}, Math. Z. \textbf{208} (1991), 567--588.
		
		\bibitem{KS} K. Schm\"udgen, {\it Trace Class Perturbations of Spectra of Self-adjoint Operators}. In: Unbounded Self-adjoint Operators on Hilbert Space. Graduate Texts in Mathematics, vol 265. (2012) Springer, Dordrecht. https://doi.org/10.1007/978-94-007-4753-1\_9
		
		\bibitem{LSJG} L. Susskind, J. Glogower, {\it Quantum Mechanical Phase and Time Operator}, Physics Vol. 1, No. 1, pp. 49--61, 1964.

        \bibitem{W15} S. Weinberg, {\it Lectures on Quantum Mechanics, Second Edition}, Cambridge University Press, 2015.
		
		\bibitem{GT} G. Teschl, \textit{Mathematical Methods in Quantuum Mechanics, with applications to Schr$\ddot{o}$dinger perators}, Amer. Math. Soc. Providence, Rhode Island, 2009.
		
	\end{thebibliography}
\end{document}